\documentclass[letterpaper, 10pt, conference]{ieeeconf}

\usepackage{graphicx}
\usepackage{subfigure}
\usepackage{stfloats}
\usepackage{multirow}

\usepackage{blindtext}
\usepackage{tasks}
\usepackage{graphicx}
\usepackage{mathrsfs}
\usepackage{amsmath,amssymb,mathrsfs,amsfonts,dsfont} 
\usepackage{breqn}
\usepackage{mathtools}
\usepackage{caption}

\def\baselinestretch{0.97}
\mathtoolsset{showonlyrefs=true}
\newtheorem{theorem}{Theorem}[section]
\newtheorem{lemma}[theorem]{Lemma}

\newtheorem{proposition}[theorem]{Proposition}

\newtheorem{definition}[theorem]{Definition}

\newtheorem{remark}[theorem]{Remark}

\newcommand{\Real}{\mathbb{R}}
\newcommand{\Natural}{\mathbb{N}}

\newcommand{\norm}[1]{\lVert#1\rVert}

\mathtoolsset{showonlyrefs=true}

\IEEEoverridecommandlockouts

\title{Compositional synthesis of almost maximally permissible safety controllers 
	\thanks{This work was supported in part by the H2020 ERC
Starting Grant AutoCPS and China Scholarship Council.
}}

\author{Siyuan Liu and Majid Zamani	
	\thanks{S. Liu is
		with the Department of Electrical and Computer Engineering, Technical University
		of Munich, Germany. M. Zamani is with the Computer Science Department, University of Colorado Boulder, USA. M. Zamani is with the Computer Science Department, Ludwig Maximilian University of Munich, Germany. Email:
		\texttt{sy.liu@tum.de}, {\tt\small majid.zamani@colorado.edu}.}
}

\begin{document}
	\maketitle
	\begin{abstract}
		In this work, we present a compositional safety controller synthesis approach for the class of discrete-time linear control systems. Here, we leverage a state-of-the-art result on the computation of robust controlled invariant sets. To tackle the complexity of controller synthesis over complex interconnected systems, this paper introduces a decentralized controller synthesis scheme. Rather than treating the interconnected system as a whole, we first design local safety controllers for each subsystem separately to enforce local safety properties, with polytopic state and input constraints as well as bounded disturbance set. Then, by composing the local controllers, the interconnected system is guaranteed to satisfy the overall safety specification. Finally, we provide a vehicular platooning example to illustrate the effectiveness of the proposed approach by solving the overall safety controller synthesis problem by computing less complex local safety controllers for subsystems and then composing them. 
	\end{abstract}
	
	\section{Introduction}
	Nowadays, there is a growing need for the controller synthesis of complex large-scale interconnected systems, e.g. autonomous vehicular control, biological networks and airplane formation flight. Treating the interconnected system in a monolithic manner is impractical owing to its inherent complexity, especially when the number of subsystems is large. Instead, compositional approaches have been developed in recent years to overcome this challenge by the so-called ``divide-and-conquer'' strategy \cite{keating2011simple}. In particular, subsystems are analyzed and tackled separately and the correctness can be ensured by the well-known ``assume-guarantee'' reasoning scheme \cite{henzinger1998you}.  
	
	Recently, there has been some research on compositionality based on the construction of (in)finite abstractions of the original control systems \cite{kim2018constructing,meyer2018compositional,tazaki2008bisimilar,zamani2017compositional,pola2016symbolic,rungger2016compositional,swikir2018dissipativity}. The abstraction is served as a substitution of the original system while designing the controller. The results presented in \cite{tazaki2008bisimilar,pola2016symbolic,rungger2016compositional} leverage the small-gain type conditions to facilitate the compositional construction of abstractions, which may not hold as the number of subsystems increases. The results proposed in \cite{zamani2017compositional,swikir2018dissipativity} take advantage of dissipativity approaches to break the requirements on the number of subsystems, whereas they restrict the interconnection topology to satisfy some graph equitability condition. The results in \cite{kim2018constructing,tazaki2008bisimilar,zamani2017compositional,pola2016symbolic,rungger2016compositional,swikir2018dissipativity} mainly deal with the compositional construction of (in)finite abstractions, whereas the ones in \cite{meyer2018compositional} additionally investigate compositional safety controller synthesis as well.  
	
	Compositional synthesis schemes without the construction of (in)finite abstractions are also presented in \cite{sadraddini2018distributed,nilsson2016synthesis}. The work in \cite{sadraddini2018distributed} studies optimal control policies  and optimal communication graphs based on distributed robust set-invariance. A synthesis of separable controlled invariant sets for physically coupled linear subsystems are developed in \cite{nilsson2016synthesis}. This result employs slack variable identities for relaxation and provides sufficient conditions based on optimization over linear matrix inequality (LMI). However, this method requires the state set to be symmetric zonotopes. Both approaches rely on some relaxed optimization problem while computing robust invariant sets, which can be conservative in terms of finding the maximally permissible safety controllers.
	
	In this work, we introduce a novel approach on decentralized synthesis of safety controllers for linear  interconnected systems by leveraging the recent result developed in \cite{rungger2017computing}. The proposed scheme is based on the so-called outer and inner invariant approximation of the maximal robust controlled invariant (RCI) set, respectively. The outer invariant approximation scheme is shown to be $\delta$-complete \cite{rungger2017computing,gao2012delta} in the sense that the algorithm either returns an empty set when the maximal invariant set is empty, or we obtain a $\delta$-relaxed RCI set (cf. Definition \ref{outercontroller}) of interest. We show that given synthesized local safety controllers, which are computed separately to enforce subsystems to satisfy local safety properties, the composed controller serves as a safety controller for the overall interconnected system. In addition, the maximality of the composed controller can be obtained using the maximality of local controllers. Similarly, the overall safety controller preserves the $\delta$-completeness property given the $\delta$-completeness property of the local safety controllers. Finally, the effectiveness of the proposed results is illustrated on a vehicular platoon example. 
	
	\section{Notation And Preliminaries} \label{section2}
	\subsection{Notation}
	
	We use $\Natural$ and $\Real$ to denote the sets of natural and real numbers, respectively. The symbols are annotated with subscripts to restrict the sets in the usual way, e.g. $\Real_{>0}$ denotes the positive real numbers. The symbol $\Real^{n\times m}$, with $n,m \in \Natural_{\geq 1}$ is used to denote the vector space of real matrices with $n$ rows and $m$ columns. We use $I_n$ and $0_{n \times m}$ to denote the identity matrix and zero matrix in $\Real^{n\times n}$ and $\Real^{n\times m}$, respectively. For $a,b \in \Real$ with $a \leq b$, the closed, open, and half-open intervals in $\Real$ are denoted by $[a,b], ]a,b[, [a,b[$, and $]a,b]$, respectively. For $a, b \in \Natural$ and $a \leq b$, the symbols $[a;b], ]a;b[, [a;b[$, and $]a;b]$ denote the corresponding intervals in $\Natural$. Given sets $X$ and $Y$, we denote by $f: X \rightarrow Y$ an ordinary map of $X$ into $Y$, whereas $f: X \rightrightarrows Y$ denotes a set-valued map.
	Given $N \in \Natural_{\geq 1}$, vectors $x_i \in \Real^{n_i}$, with $n_i \in \Natural_{\geq 1}$, $n=\sum^N_{i=1} n_i$ and $i \in [1;N]$, $x = [x_1;\ldots;x_N]$ is used to denote the concatenated vector in $\Real^n$. We use $\pi_i(x) = \Real^n \rightarrow \Real^{n_i}$ to denote the projection of vector $x$ over components $x_i$.
	We denote by $\norm{x}_2$ and $\norm{x}$, respectively, the Euclidean norm and the infinity norm of the vector $x \in \Real^n$. 	
	Given a function $f : \Real^n \rightarrow \Real^m$ and $\bar{x} \in \Real^m$, we use $f \equiv \bar{x}$ to denote that $f(x) = \bar{x}$ for all $x \in \Real^n$.	
	Given sets $X_i$, $i \in [1;N]$, the Cartesian product $X_1 \times \dots \times X_N$ is denoted by $\prod\limits_{i=1}^N X_i$. Similarly, for set $X \subseteq \Real^n$, we denote by $\pi_i(X) = 2^{\Real^n} \rightarrow 2^{\Real^{n_i}}$ the projection of $X$ over the $i$-th component.
	Given functions $f_i : X_i \rightarrow Y_i$, $i \in [1;N]$, the product function $\prod\limits_{i=1}^N f_i : \prod\limits_{i=1}^N X_i \rightarrow \prod\limits_{i=1}^N Y_i$ is defined as $\prod\limits_{i=1}^N f_i (x_1, \dots,x_N) = [f_1(x_1);\dots;f_N(x_N)]$. Notation ${\mathbb{B}}$ denotes the closed unit ball in ${\mathbb{R}^n}$ w.r.t the infinity norm $\norm{\cdot}$. Given sets $X$, $Y$ with $X \subset Y$, $Y \setminus X$ denotes the complement of $X$ with respect to $Y$, defined by $Y \setminus X = \{x : x \in Y, x \notin X \}$. The Minkowski sum for two sets $P,Q \subseteq \Real^n$ is defined by $X + Y = \{x \in \Real^n| \exists_{p \in P, q \in Q}, x = p + q\}$. Throughout the paper, for a given vector $x \in \Real^n$ and a set $W$, we slightly abuse the notation and use $x + W$ instead of $\{x\} + W$ to denote the Minkowski sum. The Pontryagin set difference is defined by $X - Y = \{x \in X |  x + Y \subseteq X\}$.

	\subsection{Interconnected Control Systems}
	First, we define control subsystems studied in this paper.
	\begin{definition}  
		A control system $\Sigma $ is a tuple
		\begin{align}
			\label{eq:controlsystem}
			\Sigma = (\Real^n,\Real^m,\Real^p,\mathcal{U},\mathcal{Z},W,f,\Real^q,h),
		\end{align}
		where $\Real^n$, $\Real^m$, $\Real^p$, $\Real^q$, are the state set, external input set, internal input set and output set, respectively. Sets $\mathcal{U}$ and $\mathcal{Z}$, respectively, are used to denote the subsets of the set of all measurable functions of time from $\Natural \rightarrow \Real^m$ and $\Natural \rightarrow \Real^p$. Set $W \subseteq \Real^n$ is an additive set of disturbances. Function $ f: \Real^n \times \Real^m \times \Real^p \times W \rightarrow \Real^n$ is the state transition function as the following: $x(t+1) = f(x(t),u(t),z(t),w(t))$, and $h : \Real^n \rightarrow \Real^q $ is the output function.	
	\end{definition}  	
	
	Now, we provide a formal definition of interconnected control systems based on the one presented in \cite{tazaki2008bisimilar}.
	We consider $N \in \Natural_{\geq 1}$ control subsystems 
	\begin{align}
		\label{eq:subcontrolsystem}
		\Sigma_i = (\Real^{n_i},\Real^{m_i},\Real^{p_i},\mathcal{U}_i,\mathcal{Z}_i,W_i,f_i,\Real^{q_i},h_i),
	\end{align}
	where $i \in [1;N]$, inputs and outputs are partitioned as
	\begin{IEEEeqnarray}{c}	
		\label{internalinput}	
		z_i = [z_{i1};\dots;z_{i(i-1)};z_{i(i+1)};\dots;z_{iN}],\\
		\label{output}	
		y_i = [y_{i1};\dots;y_{iN}],
	\end{IEEEeqnarray}
	with $z_{ij} \in \Real^{p_{ij}} $, $y_{ij} =h_{ij}(x_i) \in \Real^{q_{ij}} $ and output function
	\begin{align}
		\label{outputfunction}
		h_i(x_i) = [h_{i1}(x_i);\dots;h_{iN}(x_i)].
	\end{align}
	The outputs $y_{ii}$ are considered as external ones, whereas $y_{ij}$ with $i \neq j$ are interpreted as internal ones which are used to construct interconnections between subsystems. The dimension of $z_{ij}$ is assumed to be equal to that of $y_{ji}$. In the case that no connection exists between subsystems $\Sigma_i$ and $\Sigma_j$, we simply have $h_{ij} \equiv 0$. The interconnected control system is defined as the following.
	\begin{definition}
		\label{interconnectedsystem} 
		We consider $N \in \Natural_{\geq 1}$ control subsystems $\Sigma_i = (\Real^{n_i},\Real^{m_i},\Real^{p_i},\mathcal{U}_i,\mathcal{Z}_i,W_i,f_i,\Real^{q_i},h_i)$ as described in \eqref{internalinput}-\eqref{outputfunction}. The interconnected control system denoted by $\mathcal{I}(\Sigma_1,\dots,\Sigma_N)$ is a tuple
		\begin{align}
			\label{interConnectedsys}
			\Sigma = (\Real^n,\Real^m,\mathcal{U},W,f,\Real^q,h),
		\end{align} 
		where $n = \sum_{i=1}^{N}n_i$, $m = \sum_{i=1}^{N}m_i$, $q = \sum_{i=1}^{N}q_{ii}$, with disturbance set $W = \prod\limits_{i=1}^N W_i$, state transition function and output function  
	\begin{align}\notag
			f(x,u,w)&=[f_1(x_1,u_1,z_1,w_1);\dots;f_N(x_N,u_N,z_N,w_N)],\\ \notag
			h(x) &= [(h_{11}(x_1);\dots;h_{NN}(x_N))],
	\end{align} 
where $u = [u_1;\dots;u_N]$, $x = [x_1;\dots;x_N]$, $w = [w_1;\dots;w_N]$ and the interconnection variables are constrained by $z_{ij} = y_{ji}$ and $Y_{ji} \subseteq Z_{ij}$, $\forall i,j \in [1;N], i\neq j$.
	\end{definition}

	\subsection{Safety Controller}
	In this subsection, we define the notion of safety controller which will be used throughout the paper.
	Suppose the state constraint for each subsystem $\Sigma_i$ is given by the compact safe set $X_i \subseteq \Real^{n_i}$ and admissible input set is denoted by $U_i \subseteq \Real^{m_i} $. Then for the interconnected system $\Sigma$, the safe set $X$ and input constraint set $U$ have the following structure 
\vspace{-6mm}\begin{IEEEeqnarray}{c} \vspace{-1.5mm}
		\label{eq:safeset}
		X = \prod\limits_{i=1}^N X_i, {\kern 1pt}{\kern 1pt} {\mathop{\rm with}} {\kern 1pt}{\kern 1pt} X_i \subseteq \Real^{n_i}, \sum_{i=1}^{N}n_i = n,\\\vspace{-1.5mm}
		\label{eq:inputset}
		U = \prod\limits_{i=1}^N U_i, {\kern 1pt}{\kern 1pt}{\mathop{\rm with}} {\kern 1pt}{\kern 1pt} U_i \subseteq \Real^{m_i}, \sum_{i=1}^{N}m_i = m.
	\end{IEEEeqnarray}
	We define $Out = {\mathbb{R}^n} \setminus X $ and its projection on $\Sigma_i$ as $Out_i = {\mathbb{R}^{n_i}} \setminus X_i$. From the above structure, the state transition function of the interconnected system holds the following relations:\\
	For all $x = [x_1;\dots;x_N] \in X$, $u = [u_1;\dots;u_N] \in U$, $x' = [x'_1;\dots;x'_N] \in X$, $w = [w_1;\dots;w_N] \in W$,
		\begin{IEEEeqnarray}{c}\label{evo1}
		\begin{IEEEeqnarraybox}[\relax][l]{cc} 	
			x' = f(x,u,w) \Longleftrightarrow \forall i,j \in [1;N], i \neq j,\\
			x_{i}' = f_i(x_i,u_i,z_i,w_i), z_{ij} = h_{ji}(x_j),    
		\end{IEEEeqnarraybox}
	\end{IEEEeqnarray}	
	where $x'$ is the successor state from state $x$ under input $u$ and disturbance $w$. 
	
	For all $x = [x_1;\dots;x_N] \in X$, $u = [u_1;\dots;u_N] \in U$, $w = [w_1;\dots;w_N] \in W$,
		\begin{IEEEeqnarray}{l}\label{evo2}
		\begin{IEEEeqnarraybox}[\relax][l]{cc} 	
			&Out \cap \{f(x,u,w)\} \neq \emptyset \Longleftrightarrow \exists i \in [1;N],\forall j \in [1;N]\setminus i, \\
			&Out_i \cap \{f_i(x_i,u_i,z_i,w_i)\} \neq \emptyset, z_{ij} = h_{ji}(x_j).   
		\end{IEEEeqnarraybox}
	\end{IEEEeqnarray}	
	Now we provide the formal definition of safety controller.
	
	\begin{definition}
		\label{safetycontroller}
		A safety controller for system $\Sigma$ in \eqref{interConnectedsys} and safe set $X$ is a set-valued map $C: {\mathbb{R}^n} \rightrightarrows U$ such that:
		\begin{enumerate}
			\item $\forall x \in \Real^n$, $C(x) \subseteq U $;
			\item $dom(C) = \{ x\in {\mathbb{R}^n}|C(x) \neq \emptyset \} \subseteq X$;
			\item $\forall$ $x\in dom(C)$, $\forall u\in C(x)$, and $\forall w\in W$, $f(x,u,w) \in dom(C)$.
		\end{enumerate}	 
	\end{definition}
	\begin{remark}
		Note that the definition of a safety controller for subsystems $\Sigma_i$ is similar to that of Definition \ref{safetycontroller}, and the slight modification lies in condition 3) where the state transition function has to be modified to $f_i(x_i,u_i,z_i,w_i)$.
	\end{remark}
	
	It is known \cite{tabuada2009verification} that, there exists a maximal safety controller $C^*$ for control system $\Sigma$ and safe set $X$ containing all safety controllers, i.e., $C(x) \subseteq C^*(x)$ for all $x \in \Real^n$. This maximal safety controller can be computed theoretically using the well-known fixed-point algorithm \cite{kerrigan2001robust}.
	
	\section{Compositional Safety Controller Synthesis} \label{mainresult}
	In this section, we provide a method to compute compositionally safety controllers for interconnected system $\Sigma$ in Definition \ref{interconnectedsystem}. Suppose we are given safety controllers $C_i$ for all $i \in [1;N]$, each corresponding to subsystems $\Sigma_i$.  

	\subsection{Compositional Safety Controller Synthesis}
	Let controller $C: {\mathbb{R}^n} \rightrightarrows U$ be defined by $C(Out) = \emptyset$ and
	\begin{IEEEeqnarray}{c}\label{CompositionalController} 
		\begin{IEEEeqnarraybox}[\relax][l]{ll} 	
			&\forall i \in [1;N] {\kern 1pt}{\kern 1pt}{\kern 1pt}{\kern 1pt} {\mathop{\rm with}} {\kern 1pt}{\kern 1pt}{\kern 1pt}{\kern 1pt} x_i \in X_i,\\ \vspace{-1.5mm}
			&C(x) = \{u  \in U |u_i \in C_i(x_i), \forall i \in [1;N] \},    
		\end{IEEEeqnarraybox}
	\end{IEEEeqnarray}
	where $x = [x_1;\dots;x_N]$, $u = [u_1;\dots;u_N]$. 
	
	First, we show that the compositional controller as defined above works for the overall interconnected system. 
	\begin{theorem}
		\label{Theoremgeneral}
		Controller $C$ defined in \eqref{CompositionalController} is a safety controller for interconnected system $\Sigma$ and safe set $X$.
	\end{theorem}
	\begin{proof}
		By \eqref{CompositionalController}, $\forall x \in X $, $C(x) \subseteq U$, and $C(Out) = \emptyset \subseteq U$. Therefore, $\forall x \in {\mathbb{R}^n}$, $C(x) \subseteq U$. It is clear to see that $dom(C) \subseteq X$ trivially follows by $C(Out) = \emptyset$. We continue by showing 3) in Definition \ref{safetycontroller}. Let $x \in dom(C) \subseteq  X$, $u \in C(x)$ and $x' = f(x,u,w)$. First, we show $x' \in X$ by contradiction. If $x' \notin X$, then $x' \in Out$ from \eqref{evo2}. There exists $i \in [1;N]$, such that $Out_i \cap \{f_i(x_i,u_i,z_i,w_i)\} \neq \emptyset$, which contradicts the fact that $u_i \in C_i(x_i)$ with $C_i$ being the safety controller for subsystem $\Sigma_i$ and the corresponding safe set $X_i$. Therefore, we have $x' \in X$. From \eqref{evo1}, it is clear that, $\forall i \in [1;N]$, $x'_i = f_i(x_i,u_i,z_i,w_i)$. Moreover, $u_i \in C_i(x_i)$ implies that $x'_i \in dom(C_i)$. For $i \in [1;N]$, let $u'_i \in C_i(x'_i)$ and by \eqref{CompositionalController}, we have $u' = (u'_1, u'_2,\dots,u'_N) \in C(x')$ and $x' \in dom(C)$. Hence, we conclude that $C$ is a safety controller for $\Sigma$ and $X$.
	\end{proof}
	In the next result, we show that the maximality of the compositional controller holds when safety controller of each subsystem is maximal.

	\begin{theorem}
		\label{maximality}
		For $i\in [1;N]$, let $C_i^*$ be the maximal safety controller for subsystem $\Sigma_i$, and safe set $X_i$. Then, controller $C^*$ defined by $C^*(Out) = \emptyset$:
		\begin{IEEEeqnarray}{c}\label{maximalcontroller}
			\begin{IEEEeqnarraybox}[\relax][l]{ll} 	 	
				&\forall i \in [1;N] {\kern 1pt}{\kern 1pt}{\kern 1pt}{\kern 1pt} {\mathop{\rm with}} {\kern 1pt}{\kern 1pt}{\kern 1pt}{\kern 1pt} x_i \in X_i,\\
				&C^*(x) = \{u \in U |u_i \in C_i^*(x_i), \forall i \in [1;N] \},    
			\end{IEEEeqnarraybox}
		\end{IEEEeqnarray}
		where $x \!=\! [x_1;\!\dots\!;x_N\!]$, $u\! =\! [u_1;\!\dots\!;u_N\!]$, is the maximal safety controller for the interconnected system $\Sigma$ and safe set $X$.
	\end{theorem}
	\begin{proof}
		Let $C':  {\mathbb{R}^n} \rightrightarrows U $ be a safety controller for system $\Sigma$ and safe set $X$. For $i\in [1;N]$, let controllers  $C'_i: {\mathbb{R}^{ni}} \rightrightarrows U_i $ be defined by $C'_i(Out_i) = \emptyset $ and for all $x_i \in X_i = \pi_i(X)$,
		\begin{align}\label{proofcontroller} 
			C'_i(x_i)=\{u_i \in \pi_i(C'(x))|x \in X, x_i=\pi_i(x)\},  
		\end{align}
		where the projections for all $x \in X$ and $u \in U$ over $N$ components are $x_i=\pi_i(x) \in X_i$ and $u_i=\pi_i(x) \in U_i$.	
		We proceed with showing that $C'_i$ is a safety controller for system $\Sigma_i$ and safe set $X_i$.
		Since  $C'_i(Out_i) = \emptyset $ and $C'$ is a safety controller, for all  $x_i \in {\mathbb{R}^{n_i}}$, we readily have $C'_i(x_i)$ $\subseteq U_i$ from \eqref{proofcontroller}, so that condition 1) in Definition \ref{safetycontroller} is satisfied. For condition 2), we can trivially obtain dom($C'_i$) $ \subseteq X_i$ from $C'_i(Out_i) = \emptyset$. Now, let $x_i \in$ dom($C'_i$), $u_i \in C'_i(x_i)$, and $x_i' = f_i(x_i,u_i,z_i,w_i)$. We prove that $x_i' \in$ dom($C'_i$).  From \eqref{proofcontroller}, there exist $x \in$ dom($C'$) and $u\in C'(x)$ such that $\pi_i(x)=x_i $ and $\pi_i(x)=u_i$. As $C'$ is a safety controller, we have $f(x,u,w) \in$ dom($C'$) $\subseteq X$. Moreover, using the state transition function in \eqref{evo1}, there exists $x' = f(x,u,w)$ such that $\pi_i(x') = x_i'$. Then, since $x' \in$ dom ($C'$) and by \eqref{proofcontroller}, we get $x_i' \in$ dom($C'_i$), which satisfies condition 3) in Definition \ref{safetycontroller}. 	 	
		Therefore $C'_i$ is a safety controller for system $\Sigma_i$ and safe set $X_i$. Then, for all $x_i \in X_i$, $C'_i(x_i) \subseteq C_i^*(x_i)$ follows from the maximality of  $C_i^*$ as in Theorem \ref{maximality}. Finally, let us assume that $x \in$ dom($C'$) and $u \in C'(x)$. Then $u_i \in C'_i(x_i) \subseteq C_i^*(x_i)$, for all $i \in [1;N]$, clearly follows from \eqref{proofcontroller}. The definition of maximal safety controller in \eqref{maximalcontroller} verifies $u \in C^*(x)$, which shows the maximality of $C^*$ and completes the proof.  			
	\end{proof}
	\subsection{Safety Controller based on Robust Controlled Invariant Set}
	
	For the remainder of the paper, we make an assumption that the class of control subsystems in Definition \ref{eq:controlsystem} is linear, discrete-time and described by difference inclusion 	
	\begin{IEEEeqnarray}{c}	\hspace{-5mm}
		\label{eq:syslinear}
		\Sigma_i:\left\{ \begin{IEEEeqnarraybox}[\relax][c]{rCl}
			\xi_i(t+1) &\in& A_i\xi_i(t) + B_i\upsilon_i(t) + D_iz_i(t) + W_i,\\
			y_i(t)  &=& C_i(\xi_i(t)),
		\end{IEEEeqnarraybox}\right.
	\end{IEEEeqnarray}
	where $A_i \in \Real^{n_i\times n_i}$, $B_i \in \Real^{n_i \times m_i} $, $C_i \in \Real^{q_i \times n_i} $ and $D_i \in \Real^{n_i \times (n-n_i)}$ are constant matrices, $\xi_i : \Natural \rightarrow \Real^{n_i}$ and $y_i : \Natural \rightarrow \Real^{q_i}$ are called state trajectory and output trajectory, respectively, $\upsilon_i \in \mathcal{U}_i$ and $z_i \in \mathcal{Z}_i$ denote input trajectories. The state transition function is of the form: $f_i(x_i,u_i,z_i,w_i) = A_ix_i + B_iu_i + D_iz_i + w_i$, where $w_i \in W_i$.
	We write $\xi_{x_i\upsilon_i z_iw_i}(t)$ to denote the state value at time $t$ with initial state $\xi_i(0) = x_i$ under input trajectories $\upsilon_i$, $z_i$, and disturbance signal $w_i$. We denote by $y_{x_i\upsilon_i z_iw_i}$ the output trajectory corresponding to state trajectory $\xi_{x_i\upsilon_i z_iw_i}$.
	
	One can readily see that the safety controller in Definition \ref{safetycontroller} enforces every trajectory $\xi_{x_i\upsilon_i z_iw_i}$ of $\Sigma_i$ to evolve inside the safe set $X_i$. Note that the problem of computing this safety controller is equivalent to finding a robust controlled invariant set inside $X$, which is defined as the following.
	\begin{definition}
		A set $\Omega_i \subseteq \Real^{n_i} $ contained in $X_i \subseteq \Real^n $ is robust controlled invariant (RCI) w.r.t subsystem $\Sigma_i$ if
		\begin{align}
			\forall x_i \in \Omega_i, \exists u_i \in U_i, s.t. {\kern 1pt} {\kern 1pt} x_i' = f_i(x_i,u_i,z_i,w_i) \in \Omega_i.
		\end{align}
	\end{definition}
	Given a safe set $X$, it is known that there exists a maximal RCI set inside $X$ containing all RCI sets. It corresponds to maximal safety controller $C^*$ and is denoted by $ \Omega_\infty$. 
	However, the computation of $\Omega_\infty$, which requires implementing the well-known fixed-point algorithm, is still an open problem owing to termination and computational complexity issues. Hence, we leverage two algorithms proposed in \cite{rungger2017computing}, which are called outer and inner approximation of $\Omega_\infty$, to find RCI set for \eqref{eq:syslinear}. The main idea of the method is briefly explained here.
	
	The method is targeted at controllable linear systems of the form of $\xi(t) \in A\xi(t) + B\upsilon(t) + W$ with compact constraint sets $X$ and $U$. The computation of outer invariant approximation is based on set iteration \eqref{iter} and stopping criterion \eqref{stop} as follows
\vspace{-1.5mm}	
    \begin{IEEEeqnarray}{c}	\vspace{-0.5mm}
		\label{iter} 
		R_0 = X, R_{i+1} = pre(R_i) \cap X,	\\\vspace{-1.5mm}
		\label{stop}
		R_i \subseteq R_{i+n} + \varepsilon{\mathbb{B}}, 
	\end{IEEEeqnarray}
	where $pre(R) = \{x \in \Real^n|\exists_{u \in U}  Ax + Bu + W \subseteq R\}$ and $n$ is the dimension of the system. This method tolerates an arbitrarily small constraint violation. A $\delta-$relaxed RCI set $R$ can be provided (see \cite[Theorem 1]{rungger2017computing}), which satisfies $\Omega_\infty \subseteq R \subseteq X + \delta{\mathbb{B}}$, where $\delta = c\varepsilon$ with $c$ being a constant depending on the system dynamics. Constant $\delta$ can be interpreted as the relaxation of constraints. By choosing $\varepsilon$, $\delta$ arbitrarily small, set $R$ converges to the maximal RCI set $\Omega_\infty$.
	
	Due to the equivalence property of computing RCI set and safety controller, we provide a definition of $\delta-$relaxed safety controller w.r.t outer approximation of maximal RCI set.  
	
	\begin{definition} \label{outercontroller}
		A $\delta-$relaxed safety controller for system  $\Sigma$ and $\delta-$relaxed safe set $X+\delta {\mathbb{B}}$, based on outer approximation given by \eqref{iter} and \eqref{stop}, is a set-valued map $C^{\delta}: {\mathbb{R}^n} \rightrightarrows U+\delta {\mathbb{B}}$ such that:
		\begin{enumerate}
			\item $\forall x \in {\mathbb{R}^n} $, $C^{\delta}(x) \subseteq U+\delta {\mathbb{B}}$;
			\item $dom(C^{\delta}) = \{ x\in {\mathbb{R}^n}|C^{\delta}(x) \neq \emptyset \} \subseteq X+\delta {\mathbb{B}}$;
			\item $\forall x \in dom(C^{\delta})$ and $u\in C^{\delta}(x)$, $Ax + Bu + W \subseteq dom(C^{\delta})$.
		\end{enumerate}	
	\end{definition}
	For the inner invariant approximation, the set iteration and stopping criterion are modified to 
	\vspace{-1.5mm}	
	\begin{IEEEeqnarray}{c}	\label{iterInner} \vspace{-0.5mm}
		R_0 = X, R^{\rho}_{i+1} = pre_{\rho}(R^{\rho}_i) \cap X,	\\  \vspace{-1.5mm}
		\label{stopInner}
		R^{\rho}_i \subseteq R^{\rho}_{i+1} + \rho{\mathbb{B}}, 
	\end{IEEEeqnarray}
	where $pre_{\rho}(R^{\rho}_i) = \{x \in \Real^n|\exists_{u \in U}  Ax \!+\! Bu \!+\! W \!+\! \rho{\mathbb{B}} \subseteq R^{\rho}_i\}$. The RCI set is given by $R^{\rho}_{i+1}$ in \eqref{stopInner}, see \cite[Sec. III]{rungger2017computing}. 
	The corresponding inner safety controller is defined as follows.
	
	\begin{definition} \label{innercontroller}
		A $\rho-$inner safety controller for system  $\Sigma$ and safe set $X$, based on inner approximation given by \eqref{iterInner} and \eqref{stopInner}, is a set-valued map $C^{\rho}: {\mathbb{R}^n} \rightrightarrows U$ such that:
		\begin{enumerate}
			\item $\forall x \in {\mathbb{R}^n} $, $C^{\rho}(x) \subseteq U$;
			\item $dom(C^{\rho}) = \{ x\in {\mathbb{R}^n}|C^{\rho}(x) \neq \emptyset \} \subseteq X$;
			\item $\forall x \in dom(C^{\rho})$ and $u\in C^{\rho}(x)$, $Ax + Bu + W \subseteq dom(C^{\rho})$.
		\end{enumerate}	 
	\end{definition}

	\begin{remark}
		Note that the above method does not impose any restrictions on the shape of constraint sets or disturbance set, but simply consider compact sets which could be given by finite unions of polytopes. For the subsystems in \eqref{eq:syslinear}, there exist interconnected variables which are constrained by $z_{ij} = y_{ji}$. We basically follow an assume-guarantee reasoning \cite{henzinger1998you} to bound interconnected variables. We guarantee that the trajectory of each subsystem evolves inside its safe set under the assumption that the other $N-1$ subsystems do the same. In this view, the Minkowski sum of $D_iz_i(t) + W_i$ satisfies the compact requirement as well, so that one can take it as a disturbance set while computing the outer/inner approximation.  
	\end{remark}

	\subsection{Compositional Controller Synthesis based on Outer Approximation of Maximal Safety Controller}
	Here, we continue with compositional controller synthesis results based on separately computed $\delta-$relaxed safety controllers of subsystems.

	Let $C_i^{\delta_i}$: $\Real^{n_i} \rightrightarrows U_i +\delta_i {\mathbb{B}_i} $, $\forall i \in [1;N] $, be the $\delta_i-$relaxed safety controller for $\Sigma_i$ and the $\delta_i-$relaxed safe set $X_i+\delta_i{\mathbb{B}_i}$, where ${\mathbb{B}_i} \subseteq \Real^{n_i}$ is the closed unit ball in $\Real^{n_i}$.
	
	Then, let $C^{\delta}: {\mathbb{R}^n} \rightrightarrows U^{\delta} $ be defined by $C^{\delta}(Out^{\delta}) = \emptyset$ and 	
	\begin{IEEEeqnarray}{l}\label{outerApp}
		\begin{IEEEeqnarraybox}{l} \hspace{-5mm}
        \forall x \! \in \! X^{\delta},\! C^{\delta}\!(x) \!=\! \{u \in U^{\delta} | u_i \in C^{\delta_i}_i(x_i), \forall i \in [1;N]\},
		\end{IEEEeqnarraybox} 
	\end{IEEEeqnarray}
	where $x = [x_1;\dots;x_N]$, $u = [u_1;\dots;u_N]$, $x_i \in X_i + \delta_i{\mathbb{B}_i}$, $u_i \in U_i +\delta_i{\mathbb{B}_i}$, $X^{\delta} = \prod\limits_{i=1}^{N} (X_i + \delta_i{\mathbb{B}_i})$,  $U^{\delta} = \prod\limits_{i=1}^{N} (U_i + \delta_i{\mathbb{B}_i})$ and $Out^{\delta} = {\mathbb{R}^n} \setminus X^{\delta} $.
	
	In order to show the next result, we need the following technical lemma.
	\begin{lemma}
		\label{distributive}
		The Cartesian product of sets is distributive over Minkowski sum and Pontryagin difference.	 
	\end{lemma}
		The lemma can be proved simply by following the definitions of Minkowski sum, Pontryagin difference and Cartesian product, so that is omitted here due to space limitation.

	
	The following result shows that the compositional controller in \eqref{outerApp} works for the overall interconnected system.
	\begin{theorem}
		The controller 	$C^{\delta}$ in \eqref{outerApp} is a $\delta-$relaxed safety controller w.r.t the interconnected system $\Sigma$, $\delta-$relaxed safe set $X+\delta{\mathbb{B}}$ and $U+ \delta{\mathbb{B}} $.
	\end{theorem}
	\begin{proof}
		By Lemma \ref{distributive}, for the composed safe set $X^{\delta}$ and composed constrained input set $U^{\delta}$ in \eqref{outerApp}, we have $X^{\delta} \!=\!\prod\limits_{i=1}^{N}\! (X_i \!+\! \delta_i{\mathbb{B}_i}) \!\subseteq\! \prod\limits_{i=1}^{N}\!X_i \!+ \!\delta{\mathbb{B}} \!=\! X \!+\! \delta{\mathbb{B}}$, and $U^{\delta} \!=\!\prod\limits_{i=1}^{N}\! (U_i \!+\! \delta_i{\mathbb{B}_i}) \!\subseteq\! \prod\limits_{i=1}^{N}\!U_i \!+\! \delta{\mathbb{B}}\! =\!U \!+ \!\delta{\mathbb{B}}$, where $\delta \!=\! \norm{[\delta_1; \dots; \delta_N]}$.
		
		Rest of the proof follows the same structure as that in Theorem \ref{Theoremgeneral} and is omitted.
	\end{proof}

	\subsection{Compositional Controller Synthesis based on Inner Approximation of Maximal Safety Controller}
	In this subsection, the composed safety controller is synthesized based on inner approximation of maximal safety controller.
	
	Let $C_i^{\rho_i}$: $\Real^{n_i} \rightrightarrows U_i $, $\forall i \in [1;N] $, be the $\rho_i-$inner safety controller for $\Sigma_i$ and the safe set $X_i$.
	
	Then, let the controller $C^{\rho}: {\mathbb{R}^n} \rightrightarrows U$ be defined by $C^{\rho}(Out^{\rho}) = \emptyset$ and
		\begin{IEEEeqnarray}{l}\label{ComposedInner}
		\begin{IEEEeqnarraybox}{l} 	\hspace{-5mm}
			\forall x \!\in \!X, C^{\rho}(x) \!=\! \{u \! \in \! U |u_i \in C^{\rho_i}_i(x_i), \forall i \!\in [1;N] \}, 
		\end{IEEEeqnarraybox} 
	\end{IEEEeqnarray}
	where $x \!=\! [x_1;\dots;x_N]$, $u \!=\! [u_1;\dots;u_N]$, $x_i \in X_i$, $u_i \in U_i$.
	
	\begin{proposition}
		The set-valued map $C^{\rho}$ is a $\rho-$inner safety controller w.r.t the interconnected system $\Sigma$, safe set $X$ and $U$, where the parameter $\rho = \norm{[\rho_1; \dots; \rho_N]}$.	
	\end{proposition}
	\begin{proof}
		From \eqref{ComposedInner}, it is clear that the inner approximated controller $C^{\rho}$ follows the structure of the general safety controller in \eqref{CompositionalController}, which completes the proof.
	\end{proof}
	The next proposition shows that the controller $C^{\rho}$, with the parameters $\rho_i \in \Real_{\geq 0}, \forall i \!\in [1;N]$ suitably chosen as described in the proof, contains all the safety controllers for the interconnected system, w.r.t the deflated constraint sets $\bar{X}_{\epsilon}$ and $\bar{U}_{\epsilon}$ defined as
     \begin{IEEEeqnarray}{l} \notag \vspace{-0.5mm}
		\bar{X}_{\epsilon} =\{x \in X |x+\epsilon{\mathbb{B}} \subseteq X\},\\\notag
		\bar{U}_{\epsilon} =\{u \in U |u+\epsilon{\mathbb{B}} \subseteq U\},
	\end{IEEEeqnarray}
	where $\epsilon \in \Real_{>0}$.
	
	\begin{proposition}
		The compositional inner approximated safety controller $C^{\rho}$ for the interconnected system $\Sigma$ w.r.t $X$ and $U$ is larger than any safety controller w.r.t $\bar{X}_{\epsilon}$ and $\bar{U}_{\epsilon}$.
	\end{proposition}
	\begin{proof}
		As showed in \cite[Theorem 3]{rungger2017computing}, for each subsystem $\Sigma_i$, there exists $\rho_i \in \Real_{\geq 0}$ such that for any RCI set $\bar{R}_{\epsilon i} \subseteq \bar{X}_{\epsilon i} $ which satisfies
		\begin{align}
			x \in \bar{R}_{\epsilon i} \Longrightarrow \exists \bar{U}_{\epsilon i}: Ax + Bu + W \subseteq \bar{R}_{\epsilon i}, 
		\end{align}
		we have $\bar{R}_{\epsilon i} \subseteq R^{\rho_i}_i$. Suppose for each subsystem $\Sigma_i$, we can compute an RCI set $\bar{R}_{\epsilon i}$, which implies we have a safety controller $\bar{C}_{\epsilon i}$ w.r.t $\bar{X}_{\epsilon i}$ and $\bar{U}_{\epsilon i}$. Following the same proof idea as in Theorem \ref{Theoremgeneral}, it is readily to see that the compositional controller $\bar{C}_{\epsilon}: \Real^n \rightrightarrows \bar{U}_{\epsilon}$ defined by $\bar{C}_{\epsilon}(Out) = \emptyset$ and
		\begin{align}
			\forall x \in \bar{X}_{\epsilon}, \bar{C}_{\epsilon}(x) = \{u  \in \bar{U}_{\epsilon} |u_i \in \bar{C}_{\epsilon i}(x_i), \forall i \in [1;N] \},    
		\end{align}
		where $x = [x_1;\dots; x_N]$, $u = [u_1;\dots; u_N]$, $x_i \in \bar{X}_{\epsilon i}$, $u_i \in \bar{U}_{\epsilon i}$, is a safety controller w.r.t $\bar{X}_{\epsilon}$ and $\bar{U}_{\epsilon}$. Since $\bar{R}_{\epsilon i} \subseteq R^{\rho_i}_i$, $\forall i = [1;N]$, it follows from \cite[Theorem 3]{rungger2017computing} that any compositional safety controller $\bar{C}_{\epsilon}$ w.r.t $\bar{X}_{\epsilon}$ and $\bar{U}_{\epsilon}$ is contained in $C^{\rho}$.
	\end{proof}
	
	\section{Example} \label{example}
	We provide two case studies to illustrate our results. First, we implement the $\delta$-relaxed and  $\rho$-inner safety controllers in a platoon model to show their performance. 
	
	\begin{figure}
		\vspace{0.1cm}
		\centering
		\includegraphics[scale=0.45]{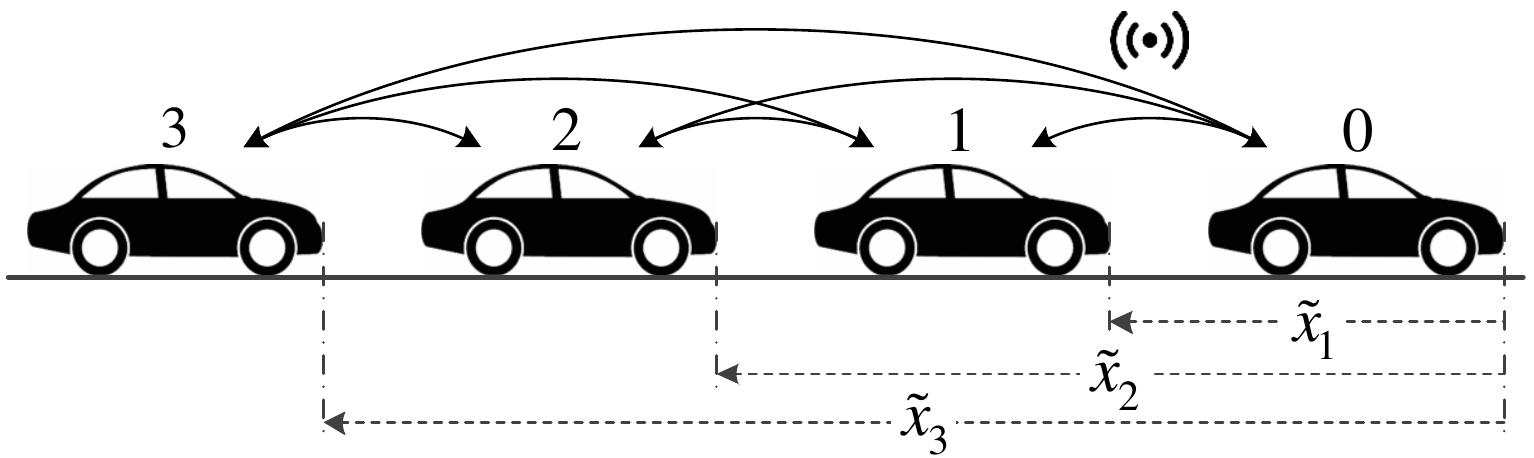}
		\caption{A platoon of 4 vehivles.}
		\label{fig:fig0}
		\vspace{-0.5cm}
	\end{figure}

	\subsection{Safety Controller for Centralized Vehicular Platoon}
	Consider a vehicular platoon example taken from \cite{sadraddini2017provably}. The platoon system consists of $N+1$ autonomous vehicles with 1 leader and $N$ followers moving on a single-lane road (see Fig. \ref{fig:fig0}). The dynamics is built in a relative manner with respect to the leader as follows
		\begin{IEEEeqnarray}{c}	
	\label{platoonmodel1}
	\begin{IEEEeqnarraybox}[\relax][c]{l} 
\tilde{x}_i(t+1) = \tilde{x}_i(t) + \tilde{v}_i(t) \Delta{\tau} + \tilde{u_i}(t)\Delta{\tau}^2/2 + w_{i,x}(t),\\
\tilde{v}_i(t+1) = \tilde{v}_i(t)  + \tilde{u_i}(t)\Delta{\tau} + w_{i,v}(t),\\
v_0(t+1) = v_0(t) + u_0(t) \Delta{\tau} + w_{0,v}(t),
	\end{IEEEeqnarraybox} 	
	\end{IEEEeqnarray}
	
	where $\tilde{x}_i(t) = x_0(t) - x_i(t)$, $\tilde{v}_i(t) = v_0(t) - v_i(t)$, $\tilde{u_i}(t) = u_0(t)-u_i(t) $, $i \in [1;N]$ denote the relative distance, relative velocity and relative input, respectively, with respect to the leader (with subscript 0). $w_{i,x}$ and $w_{i,v}$, $i \in [0;N]$ are the disturbances affecting position and velocity. The platoon state is defined as the vector $y:= (\tilde{x}_1,\tilde{v}_1,\tilde{x}_2,\tilde{v}_2,\dots,\tilde{x}_N,\tilde{v}_N,v_0)$, $y \in \Real^{2N+1}$. The evolution of state is therefore given by $y(t+1) = Ay(t) +Bu(t)+Ew(t)$, with $A$, $B$ and $E$ being constant matrices derived from dynamics in \eqref{platoonmodel1}.
	The specification that the system should satisfy includes: 1). Collision avoidance: $\tilde{x}_i > \tilde{x}_{i-1} + l_{i-1}$, $\tilde{x}_0 = 0$, where $l_i$ denotes the length of the $i$-th vehicle; 2). Constraint on the platoon length: $\tilde{x}_N \leq L$; 3). The platoon velocity is bounded: $v_0(t) \in [v_{0,\mathrm{min}}, v_{0,\mathrm{max}}]$ for all times and all admissible disturbances.
	
	Let $N=2$, $l_i=4.5$m, $\Delta{\tau}=0.5$s, $L = 10$m, $v_{0,\mathrm{min}}=13$m/s, and $v_{0,\mathrm{max}}=17$m/s. The constraints imposed on the input and disturbance are $U = \prod_{i=0}^{N}[-3,3]$m/s$^2$, and $W = \lambda^* \times \prod_{i=0}^{N}[-1,1]$m/s $\times [-0.25,0.25]$m, where the parameter $\lambda^*$ is a scalar. We compute the outer and inner approximation of the maximal RCI set. The corresponding parameters are set to $\varepsilon = \rho = 0.01$. By choosing the largest value of $\lambda^* = 0.23$, we were still able to compute the inner and outer approximation of the maximal RCI set. The projection of safe set $S$ onto the $\tilde{x}_1$-$\tilde{x}_2$ space can be given by the triangle:
	$\{(\tilde{x}_1,\tilde{x}_2)|\tilde{x}_1 + 4.5 \leq \tilde{x}_2, \tilde{x}_2 \leq 10, \tilde{x}_1 \geq 4.5\}$.
	In this case, we first used the multi-parametric toolbox \cite{MPT3} and computed the maximal RCI set. The projections of outer and inner approximation RCI sets onto the $\tilde{x}_1$-$\tilde{x}_2$ space with respect to that of the maximal one are illustrated in Fig. \ref{fig:fig1}. In comparison with the RCI set computed in \cite[Fig.2]{sadraddini2017provably}, the ones we obtained here are much less conservative. However, this centralized framework requires full state knowledge so that the computation becomes costly as $N$ increases.
	
	\begin{figure}
		\centering
		\includegraphics[scale=0.55]{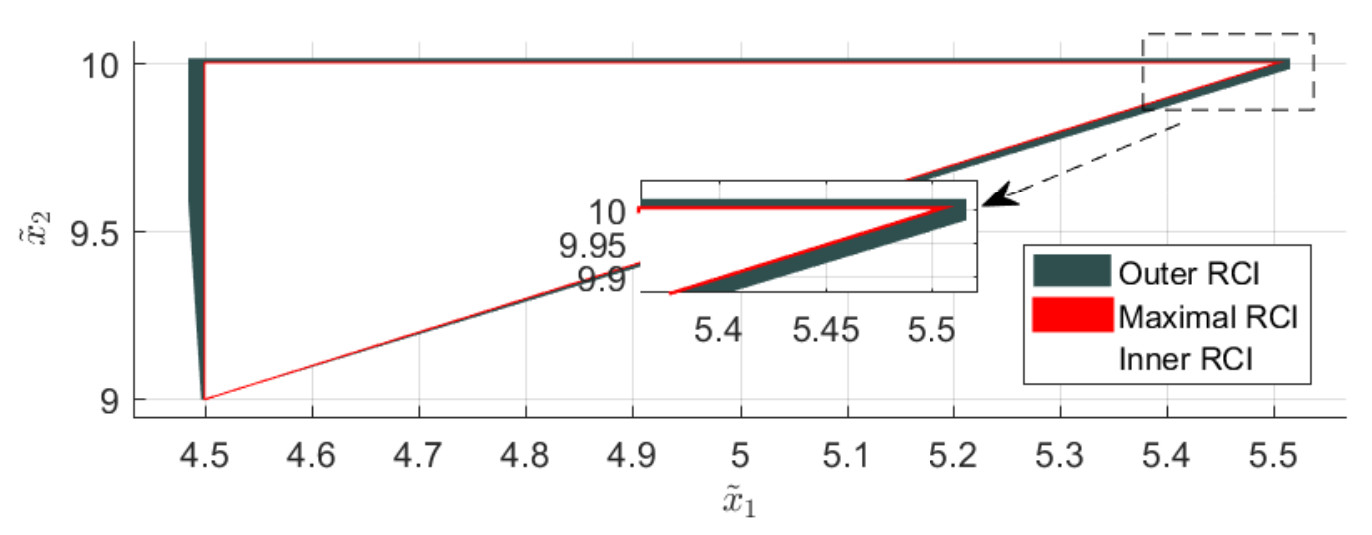}
		\caption{Projection of the outer and inner approximation of the maximal RCI set.}
		\label{fig:fig1}
		\vspace{-0.75cm}
	\end{figure}
	
	\subsection{Compositional Controller for Interconnected Platoon}
	In this example, we apply our main results to an interconnected platoon system with $N$ followers and 1 leader (see Fig. \ref{fig:fig3:b} top), as adopted from \cite{sadraddini2018distributed}. For the $i$-th follower, the state variable is defined as $x_i \!=\! (d_i,v_i)$, $i \in [1;N]$, with $d_i$ denoting the relative distance between itself (the $i$-th follower) and the preceding vehicle (the ($i-1$)-th vehicle, the 0-th vehicle represents the leader), $v_i$ is its velocity in the leader's frame. The evolution of states is given by
\begin{IEEEeqnarray}{c} \notag
x_i(t\!+\!1\!)\! =\! \begin{bmatrix}
1&-1\\
0& 1
\end{bmatrix}
\!x_i(t) \!+\!\begin{bmatrix}
0\\1
\end{bmatrix}\!u_i(t)\! +\!
\begin{bmatrix}
0 & \epsilon\\
0 & 0
\end{bmatrix} \!x_{i-1}(t)\! +\! W_i,\\ \notag
y_i = [y_{i1},\dots,y_{i,i},y_{i,i+1},\dots,y_{iN}], 
\end{IEEEeqnarray}
	where $y_{i,i} = y_{i,i+1} = x_i; y_{i,j} = 0, \forall j \neq i,i+1$, $x_{i-1}$ represents the state variable of the preceding vehicle that is acting as the interconnection here, the parameter $\epsilon$ represents the interconnection degree, $u_i(t) \in [-1,1]$ is the bounded control input, whereas $W_i$ is a polytopic disturbance set given by $\lambda^* \times [-0.1,0.1]$m $\times [-2,2]$m/s. The length of the vehicles is set to be $l_i=5$m, $\forall i \in [1;N]$.
	The overall control objective is to avoid collisions: $d_i(t) \geq 0$, $\forall i \in [1;N]$, $\forall t \in \Natural$ and in the meanwhile the length of the platoon is always bounded: $\sum_{i=1}^{N}d_i(t) + Nl_i \leq L$. Here, we decompose and under approximate the overall specification so that each vehicle is constrained by its own safe set:  $0.1 \leq d_i(t) \leq \Delta$, $\forall i \in [1;N]$, where $\Delta = (L-Nl_i)/N$. 
	
	Let $N = 6$, $\epsilon = 0.1$ and $\Delta =0.5$m. We found $\lambda^* = 0.06$ to be the largest parameter for the disturbance set while still finding a safety controller. We simulated the system for 60 seconds. The initial states were set at the centers of the safe sets. The disturbances were randomly generated at each second. The control inputs were computed using a quadratic program problem: $u_i(x_i) =$ argmin $\norm{u_i}_2$ such that $A_ix_i + B_iu_i + D_iz_i + W_i \subseteq \Omega$, where $\Omega$ denotes the outer/inner approximation of the maximal RCI set. Simulation results show that both the outer and inner compositional controller successfully enforce the interconnected system to satisfy the properties. The results of the outer approximated controller are shown in Fig. \ref{fig:fig3}. As depicted in Fig. \ref{fig:fig3:a}, the distances are always greater than $0.1$m, which indicates that the collision avoidance is ensured. The vehicular displacements are shown in Fig. \ref{fig:fig3:b}, where the horizontal axis represents relative distances of follower vehicles with respect to the leader. The leader is fixed in this frame depicted on the right by the black rectangle, and the followers depicted by gray rectangles move under disturbances.

	\begin{figure*}[!ht]
		
		\hspace{0.1cm}	
		\centering
		\subfigure[Trajectories of inputs and distances.]{
			\begin{minipage}[b]{0.3\textwidth}
				\label{fig:fig3:a}
				\includegraphics[scale=0.56]{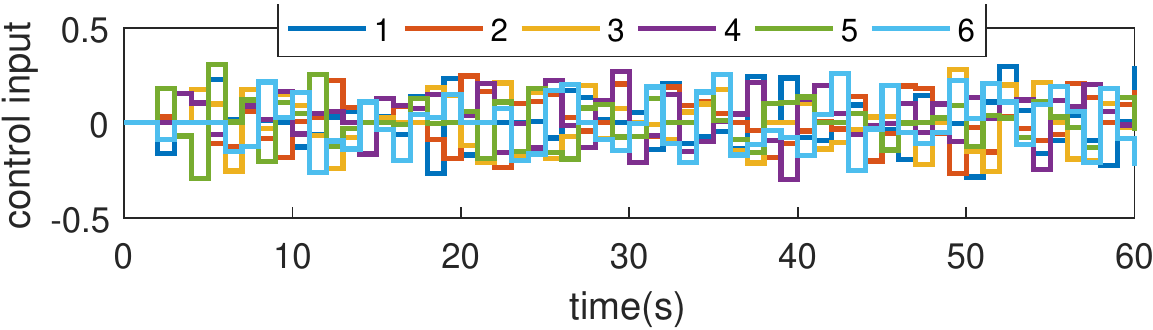}	\\						\includegraphics[scale=0.56]{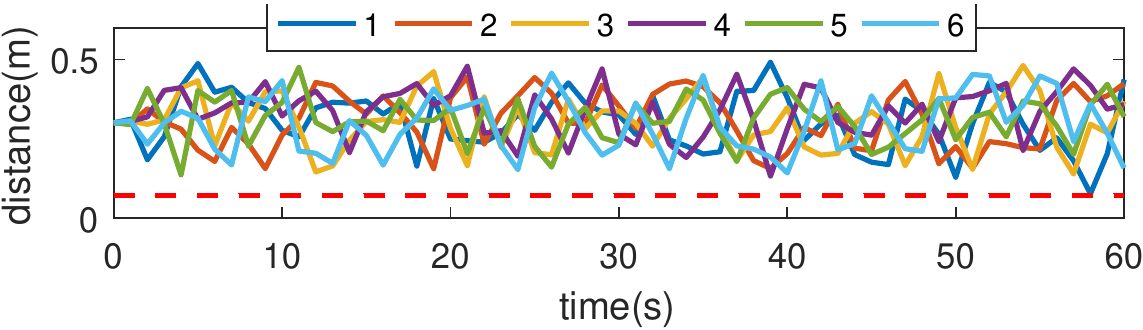}
			\end{minipage}
		}
		\subfigure[Platoon model and vehicular displacements in the leader frame.]{
			\begin{minipage}[b]{0.65\textwidth}
				\label{fig:fig3:b}
				\centering 
				\raggedleft
				\includegraphics[scale=0.36]{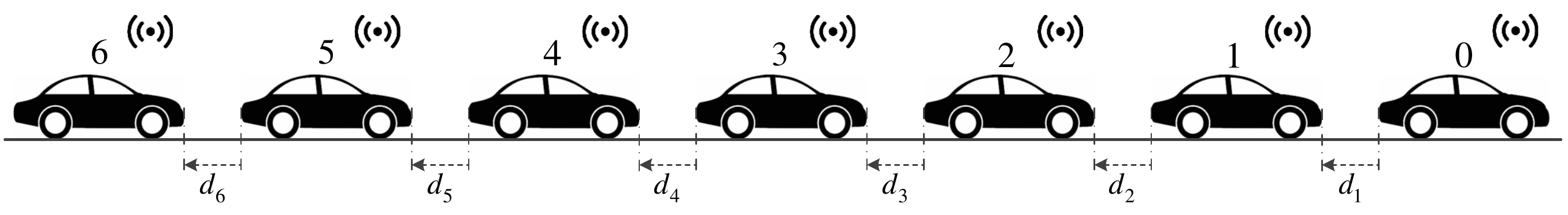}
				\label{fig:fig3:rightup}
				\vspace{0.5cm}
				\\
				\raggedleft
				\includegraphics[scale=0.63]{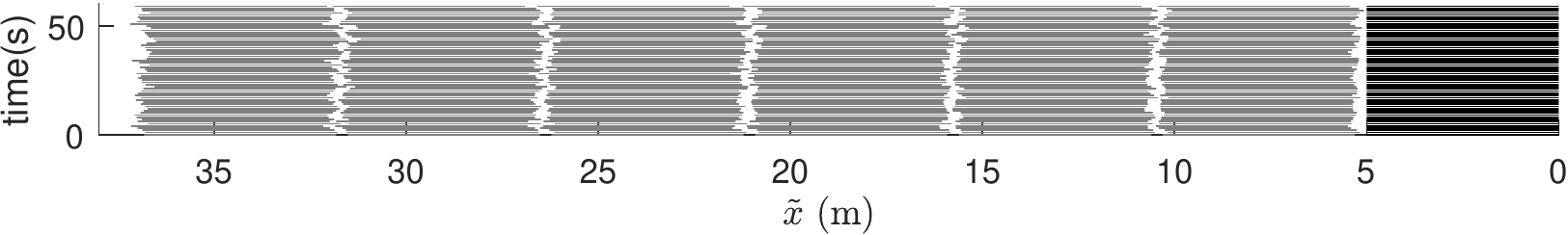}
				\label{fig:fig3:rightdown}				
			\end{minipage}
		}
		
		\small
		\caption{Trajectories of the Decentralized Platoon Model under Outer Safety Controller.}
		\label{fig:fig3}
		\normalsize
		\vspace{-0.5cm}
	\end{figure*}
Now, we analyze the largest disturbance and interconnection degree that the platoon can accommodate using outer and inner compositional controllers. The accuracy parameters of the outer and inner safety controllers are set to $\varepsilon = \rho = 0.01$. To obtain the largest disturbance magnitude, we set $\epsilon = 0.1$ and analyze how $\lambda^*$ varies with the density $\varrho = N/L$. As for the largest interconnection degree $\epsilon^*$, we set $\lambda = 0$ and see its variation with the density.
The results are shown in Table. \ref{table1}. It can be seen that $\lambda^*$ grows as vehicular density $\varrho$ decreases. Therefore, wider inter-vehicular spacings are recommended so as to adapt the system to larger disturbances. It is also observed that, the outer approximated controller can accommodate a slightly larger magnitude of disturbances than the inner approximated controller. However, the inner approximated controller slightly outperforms the outer one in terms of tolerating stronger interconnection as vehicular density goes down.
Note that the dimension of each subsystem is 2, and the control policies of this decentralized framework have $\mathcal{O}$(1) complexity. The computation time of safety controllers at each step is less than 0.01s, which is almost negligible. All the computations were conducted using MATLAB on a computer with Intel Core i7 3.4 GHz CPU.
	
	\begin{table}[t]
		\caption{Variation of $\lambda^*$ and $\epsilon^*$ with Density}
		\label{table1}
		\centering
		\begin{center}
			\begin{tabular}{|c|c||c|c|c|c|c|c|c|c|}
				\hline
				\multicolumn{2}{ |c|| } {$\Delta$} & 0.5  & 1.5 & 2 & 3 & 5 & 8 & 12 \\
				\hline
				\multicolumn{2}{ |c|| } {$\varrho$ (veh/km)} & 182 & 154 & 143 & 125 & 100 & 77 & 59 \\
				\hline
				\multirow{2}{*}{$\lambda^*$} & Outer & 0.06 & 0.18 & 0.21 & 0.26 & 0.32 & 0.37 & 0.41 \\  
				\cline{2-9}
				& Inner & 0.04 & 0.17 & 0.20 & 0.25 & 0.31 & 0.36 & 0.39 \\ 
				\hline
				\multirow{2}{*}{$\epsilon^*$} & Outer & 0.30 & 0.39 & 0.47 & 0.49 & 0.58 & 0.63 & 0.68 \\  
				\cline{2-9}
				& Inner & 0.28 & 0.39 & 0.46 & 0.49 & 0.59 & 0.64 & 0.69 \\ 
				\hline
			\end{tabular}
		\end{center}
		\vspace{-0.5cm}
	\end{table}

	\bibliographystyle{ieeetran}
	\bibliography{bib}
\end{document}